\documentclass[11pt]{article}
\pdfoutput=1

\usepackage{fullpage}

\usepackage{graphicx} 
\usepackage{subfig}

\usepackage[cmex10]{amsmath}
\usepackage{amsfonts,amssymb}
\usepackage{upgreek}

\newtheorem{lemma}{Lemma}
\newtheorem{theorem}{Theorem}
\newtheorem{corollary}[lemma]{Corollary}
\newtheorem{proposition}{Proposition}

\newenvironment{proof}[1][Proof]{\begin{trivlist}
\item[\hskip \labelsep {\bfseries #1}]}{\end{trivlist}}
\newenvironment{definition}[1][Definition]{\begin{trivlist}
\item[\hskip \labelsep {\bfseries #1}]}{\end{trivlist}}

\newenvironment{remark}[1][Remark]{\begin{trivlist}
\item[\hskip \labelsep {\bfseries #1}]}{\end{trivlist}}

\newcommand{\qed}{{\unskip\nobreak\hfil\penalty50\hskip2em\vadjust{}
           \nobreak\hfil$\Box$\parfillskip=0pt\finalhyphendemerits=0\par}}
      
\title{\LARGE\bf Improving Imprecise Compressive Sensing Models}

\begin{document}

\author{\large \sl Dongeun Lee, Rafael Lima, and Jaesik Choi\\[10pt]
\small School of Electrical and Computer Engineering\\[-0pt]
\small Ulsan National Institute of Science and Technology (UNIST), Ulsan 44919, Korea\\[-0pt]
\small eundong@unist.ac.kr, rafael\_glima@unist.ac.kr, jaesik@unist.ac.kr}
\date{}

\maketitle

\begin{abstract}
Random sampling in compressive sensing (CS) enables the compression of large amounts of input signals in an efficient manner, which is useful for many applications. CS reconstructs the compressed signals exactly with overwhelming probability when incoming data can be sparsely represented with a few components. However, the theory of CS framework including random sampling has been focused on exact recovery of signal; impreciseness in signal recovery has been neglected. This can be problematic when there is uncertainty in the number of sparse components such as signal sparsity in dynamic systems that can change over time. We present a new theoretical framework that handles uncertainty in signal recovery from the perspective of recovery success and quality. We show that the signal recovery success in our model is more accurate than the success probability analysis in the CS framework. Our model is then extended to the case where the success or failure of signal recovery can be relaxed. We represent the number of components included in signal recovery with a right-tailed distribution and focus on recovery quality. Experimental results confirm the accuracy of our model in dynamic systems.
\end{abstract}
\normalsize


\section{Introduction}
Continuous flows of big data are generated by many sources nowadays. Among these, resource limited devices occupy a significant portion. For these devices, sensing and transmitting massive data are important challenges, as they are concerned with saving resources.

Compressive sensing (CS)~\cite{bajwa2006compressive,ji2007bayesian,Seeger2008ICML,luo2009compressive,hsu2009NIPS,Lopes13,malioutov2013exact} is a well suited choice for resource limited devices because it enables the sensing and compression of massive data without the complexity burden imposed by conventional schemes. Recent advances in CS reduce the complexity burden even further with \emph{random sampling}, by which CS schemes have been successfully applied to broader application areas~\cite{baraniuk2007compressive,foucart2013mathematical,lee2014big}.

CS reconstructs the exact signals from the compressed measurements with overwhelming probability when incoming data can be sparsely represented (i.e., small numbers of components). Therefore, most CS frameworks are built based on the assumption that incoming data with sparse representation can be \emph{exactly} recovered from an enough number of measurements.

However, this assumption does not hold in practice when there is no guarantee of enough measurements for varying signal sparsity. This \emph{uncertainty} occurs especially with many dynamic systems where the numbers of components change over time. The assumption also implies that the reconstruction would fail when input signals have more components (denser) than a predefined threshold. This prevents deriving a tight probabilistic model which exploits the numbers of components and measurements in signal recovery. In this regard, recently introduced dynamic CS frameworks~\cite{sejdinovic2010Allerton,Shahrasbi2011CISS,vaswani2010modified,Ziniel2013TSP,ganguli2010NIPS,Malioutov2010JSTS} provide the way of reducing the number of necessary measurements exploiting temporal correlation between measurements. Nevertheless, a recovery success/quality analysis with uncertainty in signal sparsity has not been provided by existing CS frameworks yet.

This paper presents a new theoretical framework for the random sampling in CS that handles \emph{impreciseness} in signal recovery when the number of measurements lacks for varying signal sparsity. Our framework incorporates the beta distribution to present the signal recovery success more accurately than the success probability analysis in the CS framework. Furthermore, we relax the concept of signal recovery success and present the number of components included in the signal recovery as a varying quantity, for which we propose right-tailed distribution modeling. We believe our new framework will bridge the gap between success and failure of signal recovery in CS frameworks.

\section{Compressive Sensing and Random Sampling} \label{sec:background}
Compressive sensing, or compressed sampling (CS), is an efficient signal processing framework which incorporates signal acquisition and compression simultaneously~\cite{baraniuk2007compressive,candes2008introduction}. If a signal can be represented by only a few (significant) components with or without the help of a sparsifying basis, CS allows it to be efficiently acquired with a number of samples that is far fewer than the signal dimension and of the same order as the number of components.

\subsection{Compressing While Sensing} \label{sec:compressing}
In CS, a signal is projected onto random vectors whose cardinality is far below the dimension of the signal. Consider a signal $\mathbf{x}\in\mathbb{R}^{N}$ is compactly represented with a sparsifying basis $\mathbf{\Psi}$ having just a few components: $\mathbf{x}=\mathbf{\Psi{s}}$, where $\mathbf{s}\in\mathbb{R}^{N}$ is the vector of transformed coefficients with a few significant coefficients. Here, $\mathbf{\Psi}$ could be a basis that makes $\mathbf{x}$ sparse in a transform domain such as the DCT, wavelet transform domains, or even the canonical basis, i.e., the identity matrix $\mathbf{I}$, if $\mathbf{x}$ is sparse itself without the help of a transform.

\begin{definition} \label{def:1}
A signal $\mathbf{x}$ is called $K$-sparse if it is a linear combination of only $K {\ll} N$ basis vectors such that $\sum_{i=1}^{K}s_{n_i}{\boldsymbol{\uppsi}}_{n_i}$, where $\{n_1,{\ldots},n_K\}\subset\{1,{\ldots},N\}$; $s_{n_i}$ is a coefficient in $\mathbf{s}$; and ${\boldsymbol{\uppsi}}_{n_i}$ is a column of $\mathbf{\Psi}$.
\end{definition}

In practice, some signals may not be exactly $K$-sparse. Rather, they can be closely approximated with $K$ basis vectors by ignoring many small coefficients close to zero. This type of signal is called \emph{compressible}~\cite{baraniuk2007compressive,foucart2013mathematical}.

CS projects $\mathbf{x}$ onto a random sensing basis $\mathbf{\Phi}\in\mathbb{R}^{M\times{N}}$ as follows ($M{<}N$):
\begin{equation}
\label{eq_two}
\mathbf{y}=\mathbf{\Phi{x}}=\mathbf{\Phi\Psi{s}},
\end{equation}
where $\mathbf{\Phi}$ should have the \emph{restricted isometry property} (RIP).\footnote{The random sensing basis $\mathbf{\Phi}$ have RIP if $(1-\delta){\|\mathbf{s}\|}^{2}_{2}\leq{\|\mathbf{\Phi\Psi{s}}\|}^{2}_{2}\leq(1+\delta){\|\mathbf{s}\|}^{2}_{2}$ for small $\delta\geq0$, and this condition applies to all $K$-sparse $\mathbf{s}$.} A conventional approach for $\mathbf{\Phi}$ to satisfy RIP is sampling its independent identically distributed (i.i.d.) elements from the Gaussian or other sub-Gaussian distributions whose moment-generating function is bounded by that of the Gaussian (e.g., Rademacher/Bernoulli distribution).

The system shown in (\ref{eq_two}) is underdetermined, as the number of equations $M$ is smaller than the number of variables $N$, i.e., there are infinitely many $\mathbf{x}$'s that satisfy $\mathbf{y}=\mathbf{\Phi{x}}$. Nevertheless, this system can be solved with overwhelming probability exploiting the fact that $\mathbf{s}$ is $K$-sparse. Here $M=O(K\log(N/K))$ in the case of Gaussian and sub-Gaussian sensing matrices~\cite{candes2008introduction}.

\subsection{Random Sampling} \label{sec:random}
Random sampling is a variant of CS which can further reduce the computational complexity to a constant time~\cite{foucart2013mathematical,lee2014big}. The random sampling scheme is based on the fact that it is possible to construct $\mathbf{\Phi}$ in (\ref{eq_two}) from a random selection of rows from the identity matrix $\mathbf{I}$, which is equivalent to the random sampling of coefficients in $\mathbf{x}$.

Note that the sparsifying basis $\mathbf{\Psi}$ should be \emph{incoherent}\footnote{The two bases $\mathbf{\Phi}$ and $\mathbf{\Psi}$ are incoherent when the rows of $\mathbf{\Phi}$ cannot sparsely represent the columns of $\mathbf{\Psi}$ and vice versa.} with $\mathbf{I}$, such as the DCT and wavelet transform bases, for the successful recovery of the original signal~\cite{candes2008introduction,foucart2013mathematical}. Unless they are incoherent, the measurement vector $\mathbf{y}\in\mathbb{R}^{M}$ in (\ref{eq_two}) would contain zero entries. Here, the number of required measurements $M$ is larger than in the cases of Gaussian and sub-Gaussian matrices, that is, $M=O(K\log N)$.

\subsection{Recovery of Signal} \label{sec:recovery}
A signal recovery algorithm takes measurements $\mathbf{y}$, a random sensing matrix $\mathbf{\Phi}$, and the sparsifying basis $\mathbf{\Psi}$. The sensing matrix $\mathbf{\Phi}$ and sparsifying basis $\mathbf{\Psi}$ are assumed to be known to a decoder. The signal recovery algorithm then recovers $\mathbf{s}$ knowing that $\mathbf{s}$ is sparse. Once we recover $\mathbf{s}$, the original signal $\mathbf{x}$ can be recovered through $\mathbf{x}=\mathbf{\Psi{s}}$. The recovery algorithm reconstructs $\mathbf{s}$ by the following linear program:
\begin{equation}
\label{eq_three}
\arg\!\min{\|\mathbf{\tilde{s}}\|}_{1}\qquad\textrm{subject to}\qquad\mathbf{\Phi\Psi\tilde{s}}=\mathbf{y}.
\end{equation}
The optimization problem in (\ref{eq_three}) is solved by a $\ell_{1}$-minimization method (basis pursuit)~\cite{boyd2004convex}, greedy methods such as orthogonal matching pursuit~\cite{pati1993orthogonal}, or thresholding-based methods such as iterative hard thresholding~\cite{blumensath2008iterative}. Choosing a specific algorithm depends on $\mathbf{\Phi}$, $M$, $N$, and $K$: recovery success rates and speed can only be determined by numerical tests~\cite{foucart2013mathematical}.\footnote{Note that greedy methods are not always fast.} In this paper, we reconstruct signals by the \emph{basis pursuit}.

Specifically in the case of random sampling, the solution $\mathbf{s^{\star}}$ to (\ref{eq_three}) obeys
\begin{equation}
\label{eq_six}
{\|\mathbf{s^{\star}}-\mathbf{s}\|}_{2}\leq C_{1}\cdot{\|\mathbf{s}-{\mathbf{s}}_{K}\|}_{1}
\end{equation}
for some constant $C_{1}>0$, where ${\mathbf{s}}_{K}$ is the vector $\mathbf{s}$ with all but the largest $K$ components set to $0$. When an original signal is exactly $K$-sparse, then $\mathbf{s}={\mathbf{s}}_{K}$ with $M=O(K\log N)$ measurements, which implies that the recovery is exact, i.e., $\mathbf{s^{\star}}=\mathbf{s}$.

\section{A New Perspective on Recovery Success} \label{sec:success}
\newcommand{\Prob}{{\rm I\kern-.3em P}}
The success of signal reconstruction in compressive sensing (CS) is not deterministic. For instance, when we say an exact recovery of a $K$-sparse signal is achievable with overwhelming probability, it implies there is also the chance of recovery not being exact.

Most existing CS literature assumes a sufficient number of measurements $M$ such that an exact recovery is almost always achievable~\cite{candes2008introduction,foucart2013mathematical}, which is based on the assumption that the sparsity $K$ is already known or does not exceed a certain bound. However, the signal sparsity in dynamic systems may change over time and an excessive number of measurements may waste resources such as network bandwidth and storage space. For example, fig.~\ref{fig_one} shows recovery error over time for audio data (a 7 second recording of a trumpet solo)~\cite{ZinielRS12}, where varying signal sparsity incurs different recovery error with a fixed number of measurements over time. Here we cannot simply increase the number of measurements to eliminate error, as it is unreasonable in terms of compression. Therefore, we propose a new theoretical framework for the random sampling of CS and provide a new perspective on signal recovery.

\begin{figure}[!t]
\centering
\includegraphics[width=0.6\columnwidth]{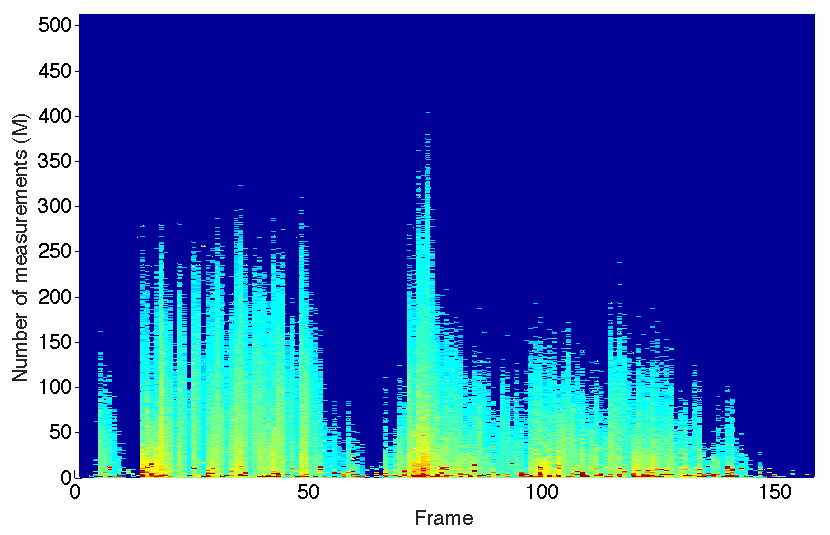}
\caption{Recovery error of audio data for different numbers of measurements $M$ over time. Each frame has a signal length $N=512$. Colors close to blue represent smaller error, whereas colors close to red represent larger error.}
\label{fig_one}
\end{figure}

\subsection{Compressive Sensing Framework} \label{sec:compressive}
In the random sampling of CS, the number of required measurements $M=O(K\log N)$ can be detailed as follows~\cite{foucart2013mathematical}:
\begin{equation}
\label{eq_eight}
M\geq C\cdot K\ln(N)\ln(\epsilon^{-1})
\end{equation}
for some constant $C>0$, where $\epsilon\in(0,1)$ denotes the probability of an \emph{inexact} recovery of the $K$-sparse signal. In particular, the signal recovery succeeds with a probability of at least $1-\epsilon$ if (\ref{eq_eight}) holds.\footnote{See Theorem 12.20~\cite{foucart2013mathematical}.}

We can then express (\ref{eq_eight}) with regard to the probability of failure $\epsilon$, which is given by
\begin{equation}
\label{eq_nine}
\Prob(\mathbf{s^{\star}}\neq\mathbf{s}\mid M,N,K):=\epsilon\leq\exp\left(-\frac{M}{C\cdot\ln(N)K}\right).
\end{equation}
Thus, the probability of failure (inexact recovery) $\Prob(\mathbf{s^{\star}}\neq\mathbf{s}\mid M,N,K)$ is conditional upon $M$, $N$, and $K$. Since we are interested in the dynamic signal sparsity $K$, we model $K$ as a \emph{random variable} with $M$ and $N$ as \emph{fixed quantities}.

If we denote an arbitrary probability density function (pdf) of $K$ as $f_{K}(k)$, we can marginalize over $k$ and find the upper bound of failure probability as follows:
\newcommand{\ud}{\mathrm{d}}
\begin{equation}
\label{eq_ten}
\Prob(\mathbf{s^{\star}}\neq\mathbf{s}\mid M,N)=\int_{k}\Prob(\mathbf{s^{\star}}\neq\mathbf{s}\mid M,N,K)\cdot f_{K}(k)\,\ud k\leq\int_{k}\exp\left(-\frac{M}{C\cdot\ln(N)K}\right)f_{K}(k)\,\ud k.
\end{equation}
Therefore, we can state that a signal recovery succeeds with a probability of at least $1-\int_{k}\exp(-M/(C\cdot\ln(N)K))f_{K}(k)\ud k$, given the distribution of signal sparsity $f_{K}(k)$.

Depending on the form of $f_{K}(k)$, the upper bound in (\ref{eq_ten}) may have an analytic solution. In particular, this is the case when $K$ follows certain distributions such as the \emph{inverse Gaussian distribution} and the \emph{gamma distribution}.\footnote{Since $K\geq0$, probability distributions supported on semi-infinite intervals, i.e., $(0,\infty)$, are rational choices.}

\begin{remark}
Assuming $f_{K}(k)=\mathrm{IG}(\mu,\lambda)$, the upper bound of (\ref{eq_ten}) is
\begin{equation}
\label{eq_inverse}
\frac{\sqrt{\lambda}\exp(\lambda/\mu-\sqrt{2\lambda/\mu^{2}}\sqrt{M/(C\cdot\ln(N))+\lambda/2})}{\sqrt{2M/(C\cdot\ln(N))+\lambda}},
\end{equation}
where $\mu$ and $\lambda$ are the mean and the shape parameter of the inverse Gaussian distribution, respectively.
\end{remark}

\begin{remark}
Assuming $f_{K}(k)=\mathrm{Gamma}(\kappa,\theta)$, the upper bound of (\ref{eq_ten}) is
\begin{equation}
\label{eq_gamma}
\frac{2}{\Gamma(\kappa)}\left(\frac{M}{C\cdot\ln(N)\theta}\right)^{\kappa/2}K_{-\kappa}\left(2\sqrt{\frac{M}{C\cdot\ln(N)\theta}}\right),
\end{equation}
where $\kappa$ and $\theta$ are the shape parameter and the scale parameter of the gamma distribution, respectively; $\Gamma(\cdot)$ is the gamma function; $K_{-\kappa}(\cdot)$ is the modified Bessel funtion of the second kind.
\end{remark}

\subsection{Modeling Success and Failure} \label{sec:beta}
Unfortunately, the probability of signal recovery failure $\epsilon$ given in (\ref{eq_nine}) does not hold in practice because there is a discrepancy between the failure probabilities in the CS framework and actual random sampling, as will be further explained in Section~\ref{sec:experimentalrecovery}. Thus we have to model the success or failure probability of signal recovery from a new perspective.

We can model the new pdf of \emph{signal recovery success} using the mixture of the \emph{Dirac delta function} and the \emph{beta distribution}, which incorporates both stochastic and deterministic cases. We introduce $K_{\mathrm{min}}$ and $K_{\mathrm{max}}$ to denote the minimum and the maximum signal sparsities which yield stochastic probability, as opposed to a deterministic result where signal recovery always succeeds or always fails.

\begin{definition} \label{def:2}
Let $\Prob(\mathbf{s^{\star}}{=}\mathbf{s}{\mid} M{,}N):=\Pi$. The pdf of $\Pi$ given $K$ is given by\footnote{$\mathrm{Beta}(\alpha_{K},\beta_{K})$ here is used to denote the pdf of the beta distribution.}
\begin{equation}
\label{eq_def1}
f_{\Pi{\mid} K}(\pi\mid k){:=}\left\{\begin{array}{ll}
\delta(\pi-1) & k<K_{\mathrm{min}}\\
\mathrm{Beta}(\alpha_{K},\beta_{K}) & \;K_{\mathrm{min}}\leq k\leq K_{\mathrm{max}}\\
\delta(\pi) & K_{\mathrm{max}}<k
\end{array} \right.\!\!\! .
\end{equation}
\end{definition}

Combining this definition with an arbitrary pdf $f_{K}(k)$ of the dynamic signal sparsity $K$, we can find the success probability distribution marginalized over $k$ as follows:
\setlength{\arraycolsep}{0.0em}
\begin{eqnarray}
\label{eq_thirteen}
f_{\Pi}(\pi)&{}={}&\int_{k}f_{\Pi\mid K}(\pi\mid k)f_{K}(k)\,\ud k\nonumber\\
&{}={}&\int_{0}^{K_{\mathrm{min}}}\delta(\pi-1)f_{K}(k)\,\ud k+\int_{K_{\mathrm{min}}}^{K_{\mathrm{max}}}\mathrm{Beta}(\alpha_{K},\beta_{K})\cdot f_{K}(k)\,\ud k+\int_{K_{\mathrm{max}}}^{\infty}\delta(\pi)f_{K}(k)\,\ud k\nonumber\\
&{}={}&\delta(\pi-1)F_{K}(K_{\mathrm{min}})+\delta(\pi)(1-F_{K}(K_{\mathrm{max}}))+\int_{K_{\mathrm{min}}}^{K_{\mathrm{max}}}\mathrm{Beta}(\alpha_{K},\beta_{K})\cdot f_{K}(k)\,\ud k,
\end{eqnarray}
where $F_{K}(\cdot)$ is the cumulative distribution function (CDF) of $K$.

The two Dirac delta function terms in (\ref{eq_thirteen}) can be interpreted as probability masses. Since $\int_{K_{\mathrm{min}}}^{K_{\mathrm{max}}}\mathrm{Beta}(\alpha_{K},\beta_{K})\cdot f_{K}(k)\,\ud k$ does not have an analytic solution, we compute the values numerically.

As an illustrative example, suppose that we examine the success probability by generating many different signed spike ($\pm1$) vectors for each signal sparsity and then performing experiments for each signed spike vector.\footnote{Detailed settings are explained in Section~\ref{sec:experimentalrecovery}.} Fig.~\ref{fig_three} shows histograms of success probability for various signal sparsities, where $K_{\mathrm{min}}=20$ and $K_{\mathrm{max}}=30$.

\begin{figure}[!t]
\centering
\includegraphics[width=0.7\columnwidth]{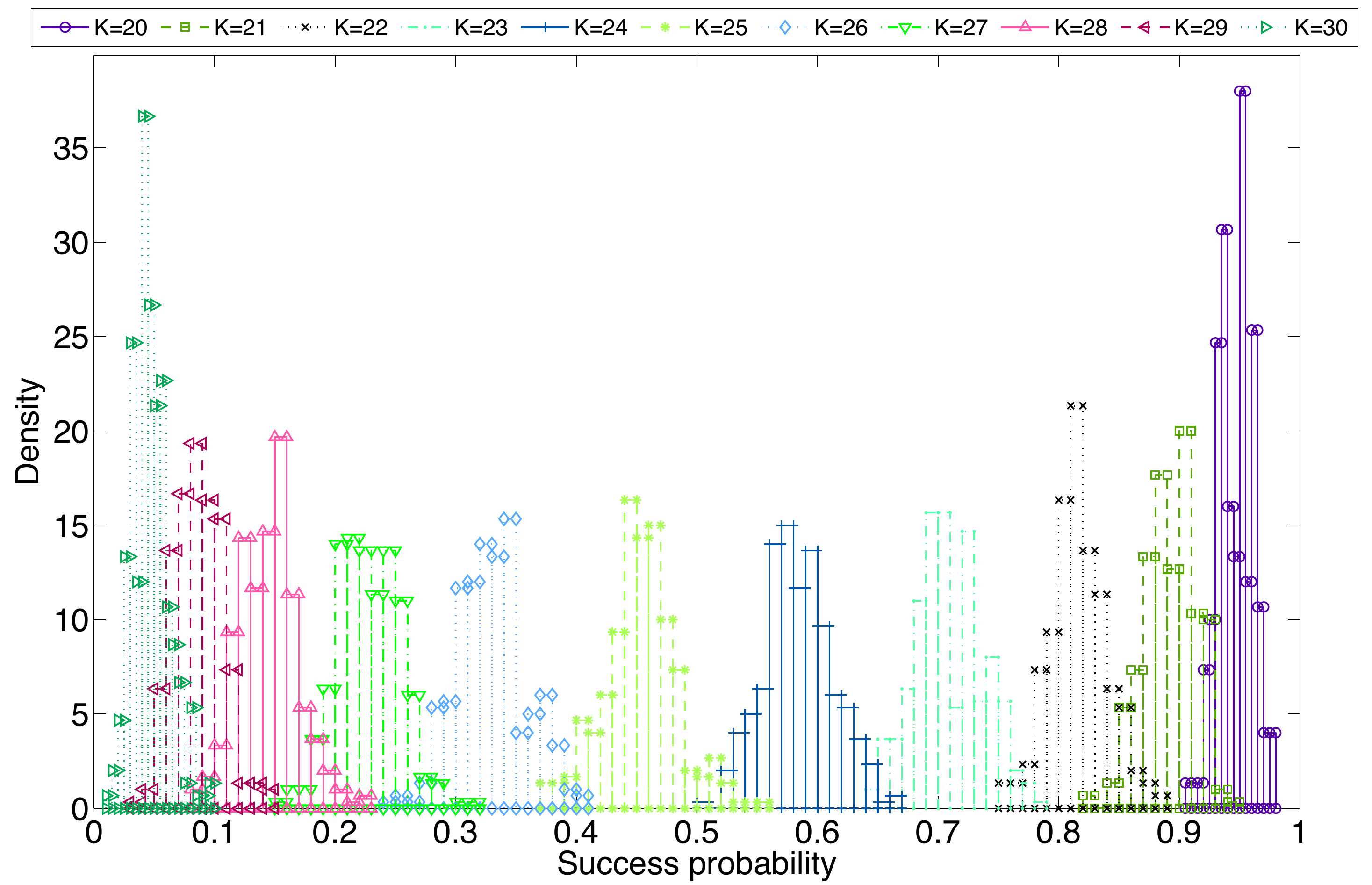}
\caption{Histograms of success probability for various $K$'s. Success probability distribution for each $K$ was obtained with 300 different random signed spike vectors for $N=512$ and $M=100$. A single success probability for each signed spike vector was calculated with 300 experiments.}
\label{fig_three}
\end{figure}

The success probability shown in Fig.~\ref{fig_three} naturally follows the beta distribution with its parameters $\alpha$ and $\beta$ depending on signal sparsity, i.e., $\Prob(\mathbf{s^{\star}}{=}\mathbf{s}\mid M,N,K)\sim\mathrm{Beta}(\alpha_{K},\beta_{K})$. The beta distribution is well known as the conjugate prior for the Bernoulli and the binomial distributions which are ideal for modeling success/failure. If more than 300 experiments had been performed in Fig.~\ref{fig_three}, the variance of each success probability distribution would have been decreased and each distribution would have been more sharply peaked.

\subsection{Modeling Accuracy} \label{sec:modeling}
Here, we present the main theoretical contribution: the recovery success model defined in (\ref{eq_def1}) is tighter than the lower bound of that in the existing CS framework explained in Section~\ref{sec:compressive}, when the number of measurements is not enough. We show the failure probability in the CS framework (\ref{eq_nine}) is incapable of reflecting the actual failure probability of signal recovery. It is not only that the inequality $\Prob(\mathbf{s^{\star}}\neq\mathbf{s}\mid M,N,K)\leq\exp(-M/(C\cdot\ln(N)K))$ cannot provide tight probability of failure, but the inequality itself is inaccurate.

This inaccuracy results from the slowly decaying lower bound of success probability, that is, $1-\exp(-M/(C\cdot\ln(N)K))$. In fact, we can show this lower bound decays slower than a power-law decay by the following lemma.

\begin{lemma}[Slackness of Recovery Success Probability] \label{lem:1}
There exists $K_{0}>0$ such that for all $K>K_{0}$, the lower bound of recovery success probability in the CS framework (Section~\ref{sec:compressive}) is greater than the value of a power-law-decay function.
\end{lemma}
\begin{proof}
We need to show the following inequality
\begin{equation}
\label{eq_lem1}
1-\exp\left(-\frac{M}{C\cdot\ln(N)K}\right)>K^{-\alpha}
\end{equation}
holds if $K>K_{0}$ for some $K_{0}>0$, where $\alpha>0$. Adding, subtracting, and taking the power $K$ on both sides yields
\begin{equation}
\label{eq_lem2}
(1-K^{-\alpha})^K>\exp\left(-\frac{M}{C\cdot\ln(N)}\right).
\end{equation}

We now use the \emph{binomial approximation} on the left-hand side: $(1-K^{-\alpha})^K\geq 1-K\cdot K^{-\alpha}$. Thus we instead prove the following inequality
\begin{equation}
\label{eq_lem3}
1-K^{1-\alpha}>\exp\left(-\frac{M}{C\cdot\ln(N)}\right).
\end{equation}
holds if $K>K_{0}$ for some $K_{0}>0$.

If we assume $\alpha>1$, then adding, subtracting, and taking the power $1/(1-\alpha)$ on both sides of (\ref{eq_lem3}) yields
\begin{equation}
\label{eq_lem4}
\left(1-\exp\left(-\frac{M}{C\cdot\ln(N)}\right)\right)^{1/(1-\alpha)}<K.
\end{equation}
Setting $K_{0}=(1-\exp(-M/(C\cdot\ln(N))))^{1/(1-\alpha)}$, we can argue that for all $K>K_{0}$, the lower bound of recovery success probability is greater than the value of a power-law-decay function.\qed
\end{proof}

\begin{corollary} \label{cor:1}
In the CS framework (Section~\ref{sec:compressive}), there is always a chance of succeeding at signal recovery however large $K$ is.
\end{corollary}
\begin{proof}
The power-law-decay function $K^{-\alpha}$ in (\ref{eq_lem1}) \emph{slowly} converges to zero as $K\to\infty$: its value is noticeably greater than zero even with large $K$. As the lower bound of recovery success probability is greater than the value of the power-law-decay function for all $K>K_{0}$, we can say there is always a chance of recovery success however large $K$ is.\qed
\end{proof}

We can now show that our recovery success model provides more accurate success probability by the following theorem. 
\begin{theorem} \label{thm:1}
The recovery success model in (\ref{eq_def1}) is tighter than the lower bound of recovery success probability given by the CS framework (Section~\ref{sec:compressive}) with a limited number of measurements.
\end{theorem}
\begin{proof}
The claim of the CS framework in Corollary~\ref{cor:1} is in fact implausible because it says we can even set $K>M$ and there is still a chance of success. We cannot expect signal recovery with a number of measurements $M$ less than $K$.

On the contrary, our recovery success model can yield $\Prob(\mathbf{s^{\star}}=\mathbf{s}\mid M,N,K)=0$ with a bounded $K_{\mathrm{max}}$. In particular, we can let the mean of $\mathrm{Beta}(\alpha_{K},\beta_{K})$, $\alpha_{K}/(\alpha_{K}+\beta_{K})$, converge to zero with $\alpha_{K_{\mathrm{max}}}\to 0$.

Similarly, we show this mean converges to one ($\Prob(\mathbf{s^{\star}}=\mathbf{s}\mid M,N,K)=1$) with $K_{\mathrm{min}}$ which is not so close to zero, whereas the lower bound of the recovery success probability given by the CS framework converges to one only if $K$ is very close to zero.

We can let $\alpha_{K}/(\alpha_{K}+\beta_{K})$ converge to one with $\beta_{K_{\mathrm{min}}}\to 0$. In contrast, $1-\exp(-M/(C\cdot\ln(N)K))\to 1$ if, and only if, $K\to 0$. Since $0<K_{\mathrm{min}}<K_{\mathrm{max}}<\infty$, we can argue that our recovery success model can provide tighter recovery success probability.\qed
\end{proof}

\subsection{Parameter Learning in Dynamic Systems} \label{sec:hmm}
When the signal sparsity $K$ changes in dynamic systems, it does not change in an abrupt manner; rather, it tends to smoothly change over time~\cite{vaswani2010modified,Ziniel2013TSP}. One simple way to model this correlation between $K$'s is to utilize the \emph{Markov model}~\cite{Ziniel2013TSP}. Here, each $K$ makes up a state and each state is associated with the recovery success probability. This can be best modeled by the hidden Markov model, where each state $K$ generates success/failure according to the emission probability.

In our scenario, signal recovery success is observed in an environment where the signal sparsity varies over time. We want to estimate parameters of the hidden Markov model, especially the emission probabilities. Since our recovery success model employs the beta distribution as conjugate distributions (prior and posterior), we can learn its parameters $\alpha_{K}$ and $\beta_{K}$ for each state $K$.

Specifically, the decoder can observe signal recovery success/failure and corresponding signal sparsity $K$ at each decoding step. Then using these emission and state sequences, it can sequentially update the parameters $\alpha_{K}$ and $\beta_{K}$ for each state $K$~\cite{durbin1998biological}. In order to prevent over-fitting with insufficient observations, it is preferrable to have hyperparameters set according to $K$'s. In Fig.~\ref{fig_three}, we can clearly see the trend of $\alpha_{K}$ and $\beta_{K}$ for different $K$'s: $\alpha_{K}$ decreases, whereas $\beta_{K}$ increases as $K$ grows. (Also see the proof of Theorem~\ref{thm:1}.)

\section{Further Analysis on Recovery Quality} \label{sec:quality}
When a signal of interest is not exactly $K$-sparse but compressible, as discussed in Section~\ref{sec:compressing}, the signal recovery in Section~\ref{sec:recovery} can be treated from a different perspective~\cite{lee2015learning}. In particular, the inequality (\ref{eq_six}) is considered differently.

If an original signal is compressible, then the quality of a recovered signal is proportional to that of the $K$ most significant pieces of information. We get progressively better results as we compute more measurements $M$, since $M=O(K\log N)$~\cite{candes2008introduction}. Therefore, $\mathbf{\Psi{s^{\star}}}\in\mathbb{R}^{N}$ also makes progress on its quality as $M$ increases.\footnote{The error bound follows (\ref{eq_six}) as well if $\mathbf{\Psi}$ is an orthogonal matrix, which is usually the case.}

From this viewpoint, the success or failure of signal recovery no longer exists. Rather, we can view the number of components included in the signal recovery as a \emph{varying quantity}. Specifically, if a signal recovery is about to fail with a given $K$, then $K$ can be lowered to make the recovery eventually succeed. Here the number of \emph{included} components $K$ varies for different recoveries and signals, as analogous to the success probability in Section~\ref{sec:beta} that can be calculated with different recoveries and varies for different signals.

In this regard, (\ref{eq_six}) can be utilized to infer varying $K$'s over different recoveries and signals. Here our assumption is that the upper bound in (\ref{eq_six}) is \emph{tight} such that we solve the following optimization problem:
\begin{equation}
\label{eq_fourteen}
\max{K}\qquad\textrm{subject to}\qquad{\|\mathbf{s^{\star}}-\mathbf{s}\|}_{2}\leq C_{1}\cdot{\|\mathbf{s}-{\mathbf{s}}_{K}\|}_{1}.
\end{equation}
In (\ref{eq_fourteen}), $C_{1}$ has to be determined, where the maximum signal sparsity $K_{\mathrm{max}}$ introduced in Section~\ref{sec:beta} plays a key role to set the upper limit on how large $K$ can be, since $K>K_{\mathrm{max}}$ is not reasonable.

In particular, we can generate a compressible signal $\mathbf{s}_i\in S$ such that ${\|\mathbf{s}_i\|}_{1}=C_{\ell_{1}}$ and ${\|\mathbf{s}_i\|}_{2}=C_{\ell_{2}}$ for all $i$, where $S$ is the set containing many different signals; $C_{\ell_{1}}>0$ and $C_{\ell_{2}}>0$ being constants. For each $\mathbf{s}_i$, we have a set $S_{i}^{\star}$ which contains many different recoveries $\mathbf{s}_{ij}^{\star}$. Then $C_{1}$ can be found as follows:
\begin{equation}
\label{eq_fifteen}
C_{1}=\frac{\min{\|\mathbf{s}_{ij}^{\star}-\mathbf{s}_i\|}_{2}}{{\|\mathbf{s}_i-\mathbf{s}_{i}^{K_{\mathrm{max}}}\|}_{1}},
\end{equation}
where $\mathbf{s}_{i}^{K_{\mathrm{max}}}$ denotes the compressible signal $\mathbf{s}_i$ with all but the largest $K_{\mathrm{max}}$ components set to $0$.

Varying $K$'s obtained through (\ref{eq_fourteen}) can be represented by a pdf, which has been empirically shown to follow the \emph{gamma distribution}~\cite{lee2015learning}. We are interested in the shape of this pdf, which is shown by the following proposition.

\begin{proposition} \label{pro:1}
The pdf of $K$, the number of components included in the signal recovery of a compressible signal, is skewed to the right, i.e., right tailed.
\end{proposition}
\begin{proof}
Since ${\|\mathbf{s}_i\|}_{1}=C_{\ell_{1}}$ and ${\|\mathbf{s}_i\|}_{2}=C_{\ell_{2}}$ for all $i$, we can conceive the same sequence $\{s_n\}$ of elements (absolute values) in $\mathbf{s}_i$ for all $i$. Then we have
\begin{equation}
\label{eq_pro1}
{\|\mathbf{s}_i-\mathbf{s}_{i}^{K}\|}_{1}=\sum_{n=1}^{N-K}s_n.
\end{equation}
Without loss of generality, we consider the partial sum $\sum_{n=1}^{N-K}s_n$ in (\ref{eq_pro1}) to be an \emph{arithmetic series} which can be represented by a quadratic function in terms of $K$. We also assume the inequality constraint in (\ref{eq_fourteen}) is the equality constraint such that ${\|\mathbf{s^{\star}}-\mathbf{s}\|}_{2}=C_{1}\cdot{\|\mathbf{s}-{\mathbf{s}}_{K}\|}_{1}$.

If we take the (partial) inverse function of the quadratic function, we have $K\sim K_{\mathrm{max}}-\sqrt{{\|\mathbf{s^{\star}}-\mathbf{s}\|}_{2}-(\min{\|\mathbf{s}_{ij}^{\star}-\mathbf{s}_i\|}_{2})}$. Assuming the distribution of ${\|\mathbf{s^{\star}}-\mathbf{s}\|}_{2}$ is \emph{symmetric} (zero skewness), this asymptotic relation says ${\|\mathbf{s^{\star}}-\mathbf{s}\|}_{2}$ will be \emph{compressed} as it becomes large, which in turn makes the pdf of $K$ right tailed.

A similar claim can be made if we consider the partial sum $\sum_{n=1}^{N-K}s_n$ to be a \emph{geometric series}, where $K\sim N-\log({\|\mathbf{s^{\star}}-\mathbf{s}\|}_{2})$. In this case, the pdf of $K$ is skewed to the right as well.\qed
\end{proof}

\subsection{Error Analysis in Dynamic Systems} \label{sec:dynamic}
Since the success or failure of signal recovery does not exist in this framework, we instead investigate the amount of error occurring during the recovery procedure in an \emph{expected value} sense. In particular, the best $K$-term approximation ${\|\mathbf{s}-{\mathbf{s}}_{K}\|}_{1}$ in (\ref{eq_six}) is known to be bounded as follows~\cite{baraniuk2010model}:
\begin{equation}
\label{eq_sixteen}
{\|\mathbf{s}-{\mathbf{s}}_{K}\|}_{1}\leq\frac{2G}{K},
\end{equation}
where the constant $G$ can be learned by the power-law decay such that each magnitude of components in $\mathbf{s}$, sorted in decreasing order, is upper bounded by $G/i^{2}$. ($i=1,\ldots,N$ is the sorted index.)

Then we can analyze the $\ell_{2}$ error $E$ of signal recovery assuming $f_{K}(k)=\mathrm{Gamma}(\kappa,\theta)$, which is given by
\begin{equation}
\label{eq_seventeen}
E=\int_{k}C_{1}\cdot\frac{2G}{k}f_{K}(k)\,\ud k=\frac{2C_{1}G}{\theta}\mathrm{B}(\kappa-1,1),
\end{equation}
where $\mathrm{B}(\cdot,\cdot)$ is the beta function~\cite{lee2015learning}. Here the pdf $f_{K}(k)$ is employed to represent varying $K$'s.\footnote{Note that this pdf is different from the one introduced in Section~\ref{sec:success}.}

In this framework, there is no longer such an indicator as the timely varying signal sparsity $K$ in Section~\ref{sec:success}, because signals are compressible and their coefficients are already populated with small, but non-zero, coefficients. Thus, we may assume the same gamma distribution over time, whose parameters $\kappa$ and $\theta$ can then be estimated.

In order to prevent overfitting to insufficient observations, we introduce the conjugate prior for the gamma distribution. It is known that the conjugate prior of the gamma distribution has the following form~\cite{miller1980bayesian,Fink97acompendium}.
\begin{equation}
\label{eq_eighteen}
\Prob(\kappa,\theta\mid p,q,r,s)=\frac{1}{Z}\cdot\frac{p^{\kappa-1}\exp(-q/\theta)}{\Gamma(\kappa)^{r}\theta^{s\kappa}},
\end{equation}
where $p$, $q$, $r$, and $s$ are hyperparameters which are sequentially updated with $p'=pk$, $q'=q+k$, $r'=r+1$, and $s'=s+1$, respectively\footnote{Here, $p'$, $q'$, $r'$, and $s'$ are updated posterior hyperparameters; $k$ is a single observation.}; and the normalizing constant $Z$ is
\begin{equation}
\label{eq_nineteen}
Z=\int_{0}^{\infty}\frac{p^{\kappa-1}\Gamma(s\kappa+1)}{\Gamma(\kappa)^{r}q^{s\kappa+1}}\,\ud\kappa.
\end{equation}

Using (\ref{eq_seventeen}) and (\ref{eq_eighteen}), we can marginalize over $\kappa$ and $\theta$ to estimate error $\widehat{E}$ as follows:
\begin{eqnarray}
\label{eq_twenty}
\widehat{E}&{}={}&\int_{\kappa}\int_{\theta}E\cdot\Prob(\kappa,\theta\mid p,q,r,s)\,\ud\theta\,\ud\kappa\nonumber\\
&{}={}&\frac{2C_{1}G}{Z}\int_{0}^{\infty}\frac{p^{\kappa-1}}{(\kappa-1)\Gamma(\kappa)^{r}}\int_{0}^{\infty}\frac{\exp(-q/\theta)}{\theta^{s\kappa+1}}\,\ud\theta\,\ud\kappa\nonumber\\
&{}={}&\frac{2C_{1}G}{Z}\int_{0}^{\infty}\frac{p^{\kappa-1}\Gamma(s\kappa)}{(\kappa-1)\Gamma(\kappa)^{r}q^{s\kappa}}\,\ud\kappa,
\end{eqnarray}
\setlength{\arraycolsep}{5pt}%
which can be computed numerically.

\section{Experimental Results} \label{sec:experimental}
\subsection{Recovery Success} \label{sec:experimentalrecovery}
In Section~\ref{sec:success}, we discussed the discrepancy between the failure probabilities in the CS framework and actual random sampling. In order to show this discrepancy, we \emph{artificially} generated signed spikes $\pm1$ at random locations in proportion to desired sparsities and \emph{densified} these spikes using $\mathbf{\Psi}$\footnote{We used DCT as the sparsifying basis $\mathbf{\Psi}$ throughout experiments.} to perform the random sampling.

For each signal sparsity $K$, the actual failure probability can be calculated for different recovery experiments. To this end, we adopted a standard optimization method (\emph{basis pursuit}) to solve the optimization problem in (\ref{eq_three})~\cite{chen1998atomic}. Specifically, the primal-dual algorithm based on the interior point method was employed to solve (\ref{eq_three})~\cite{boyd2004convex}.

Fig.~\ref{fig_two} shows that the actual failure probability of signal recovery with varying signal sparsity does not follow the failure probability given in the CS framework. The failure probability in (\ref{eq_nine}) cannot model the actual failure probability of signal recovery, regardless of the value chosen for constant $C$. This result confirms Lemma~\ref{lem:1} and Corollary~\ref{cor:1}.

\begin{figure}[!t]
\centering
\includegraphics[width=0.6\columnwidth]{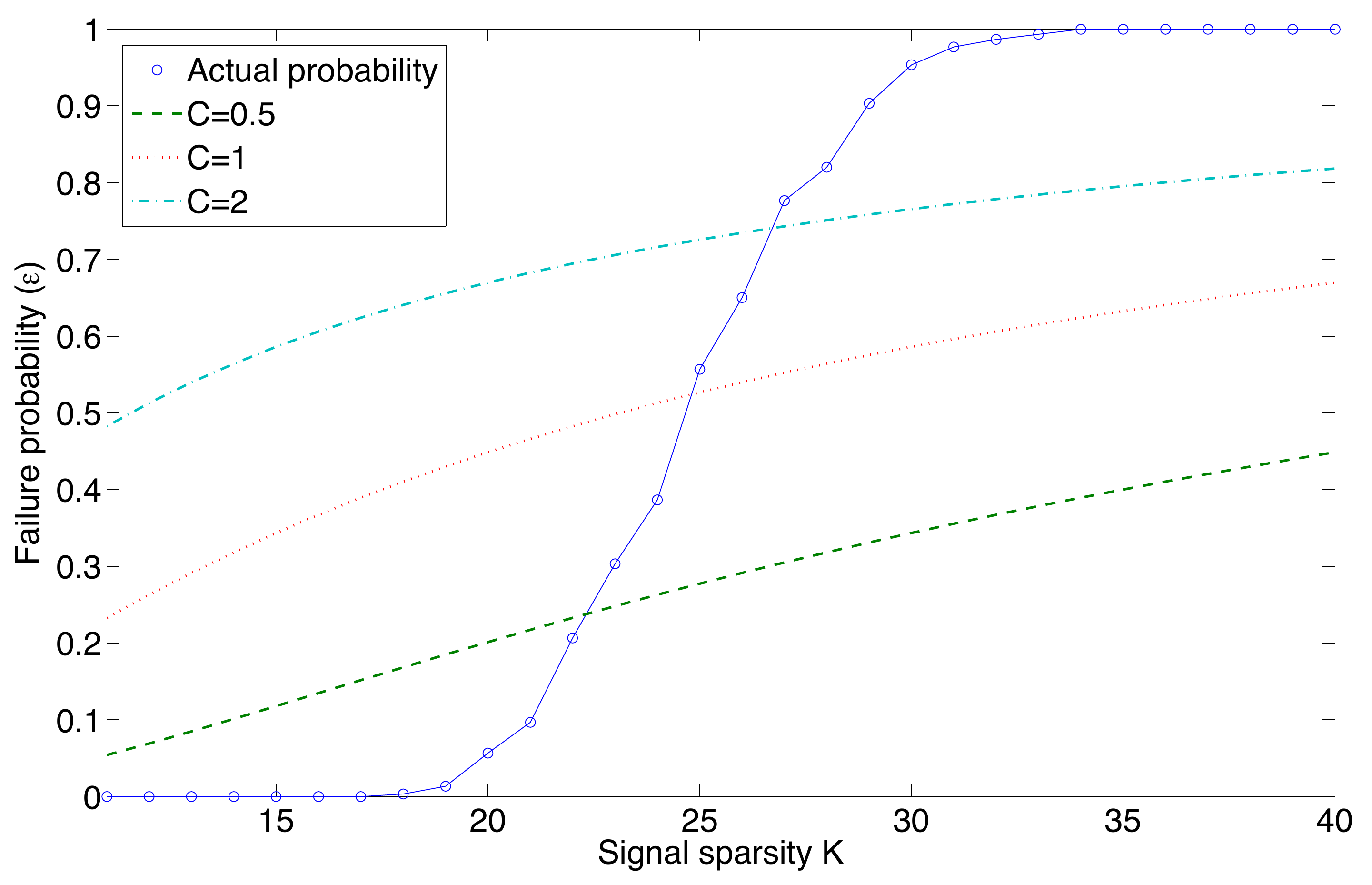}
\caption{Comparison between actual failure probability and failure probabilities given in (\ref{eq_nine}) with varying $C$'s. Actual failure probability of each signal sparsity $K$ was obtained with 300 experiments for $N=512$ and $M=100$.}
\label{fig_two}
\end{figure}

Moreover, in Section~\ref{sec:beta} we modeled the new pdf of signal recovery success $f_{\Pi}(\pi)$ in (\ref{eq_thirteen}). We compared this new pdf with the upper bound of failure in (\ref{eq_ten}), given a dynamic signal sparsity $K$. Specifically, we employed the inverse Gaussian distribution such that $f_{K}(k)=\mathrm{IG}(30,200)$. Fig.~\ref{fig_efficacy} exhibits the efficacy of our recovery success model, where the lower bounds of success probability given in the CS framework fail to capture actual success probability in random sampling case. This result confirms Theorem~\ref{thm:1}.

\begin{figure}[!t]
\centering
\includegraphics[width=0.6\columnwidth]{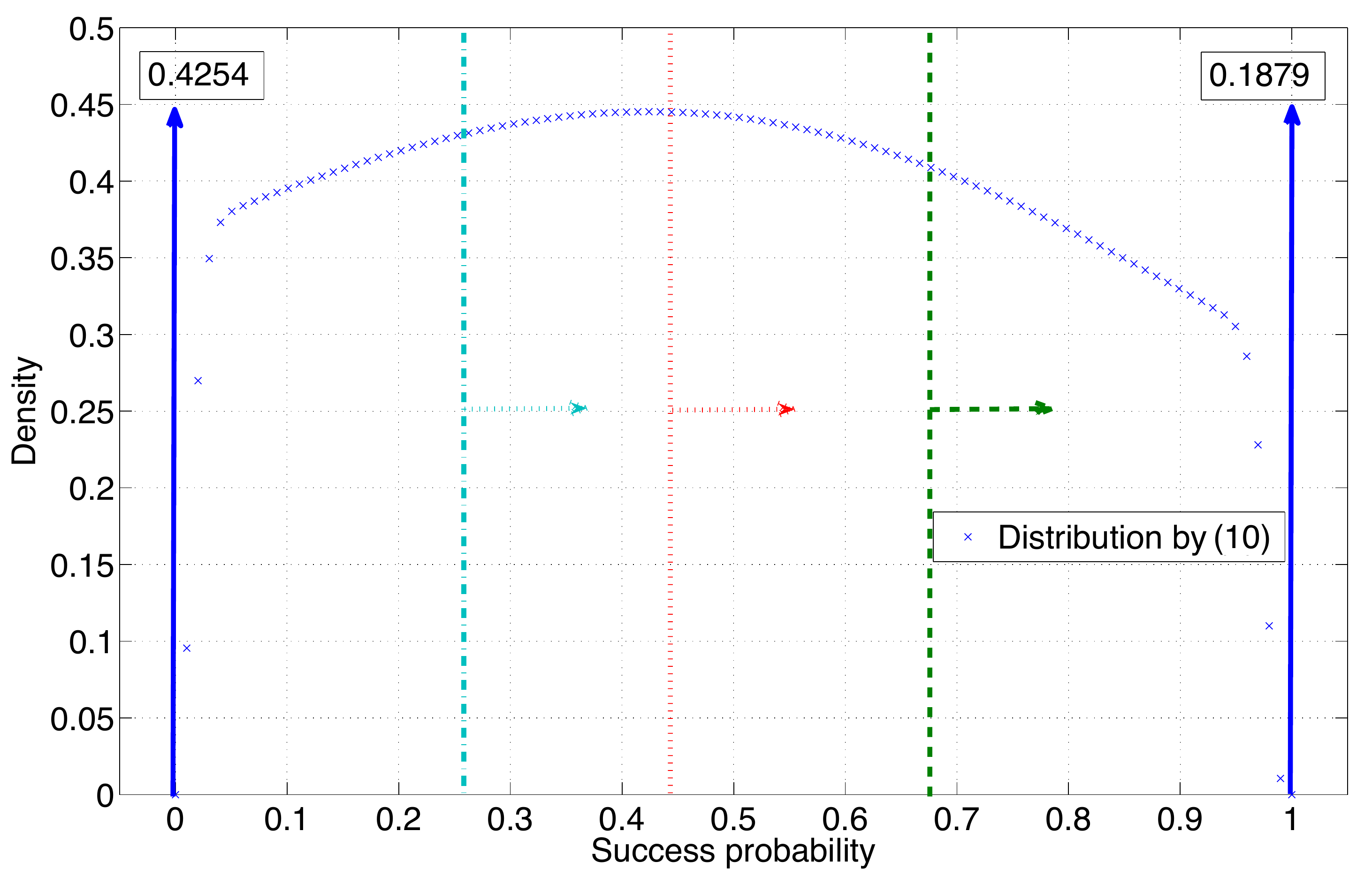}
\caption{Comparison between our new success probability distribution in (\ref{eq_thirteen}) and the lower bounds of success probability obtained by (\ref{eq_ten}) with varying $C$'s. The inverse Gaussian distribution was used for $f_{K}(k)$. Two probability masses are shown by vertical arrows, where solid boxes atop the arrows denote their probabilities. Three vertical dashed/dotted lines represent the lower bounds by (\ref{eq_ten}): $C=0.5$ at $0.6781$; $C=1$ at $0.4450$; and $C=2$ at $0.2596$. Here, $K_{\mathrm{min}}=20$ and $K_{\mathrm{max}}=30$.}
\label{fig_efficacy}
\end{figure}

Note that our recovery success model provides the baseline of recovery success for any CS frameworks that are specifically designed to handle varying signal sparsity. For instance, Fig.~\ref{fig_modified} shows histograms of success probability for various signal sparsities using Modified-CS~\cite{vaswani2010modified}.\footnote{Results were obtained with two frames where the second frame has one more spike than the first frame so that Modified-CS could exploit smoothly varying signal sparsity. Histograms in Fig.~\ref{fig_modified} are the success probability of the second frame.} Compared with Fig.~\ref{fig_three}, the success probability shown in Fig.~\ref{fig_modified} also follows the beta distribution; but success probability is higher than that of basis pursuit for a given sparsity $K$ ($K_{\mathrm{min}}=21$ and $K_{\mathrm{max}}=31$), thanks to the ability of Modified-CS to handle dynamic signal sparsity. The recovery success model in (\ref{eq_def1}) is still effective here for a theoretical framework, or the recovery success model using basis pursuit may promise a minimum guarantee for the recovery success of other CS frameworks.

\begin{figure}[!t]
\centering
\includegraphics[width=0.7\columnwidth]{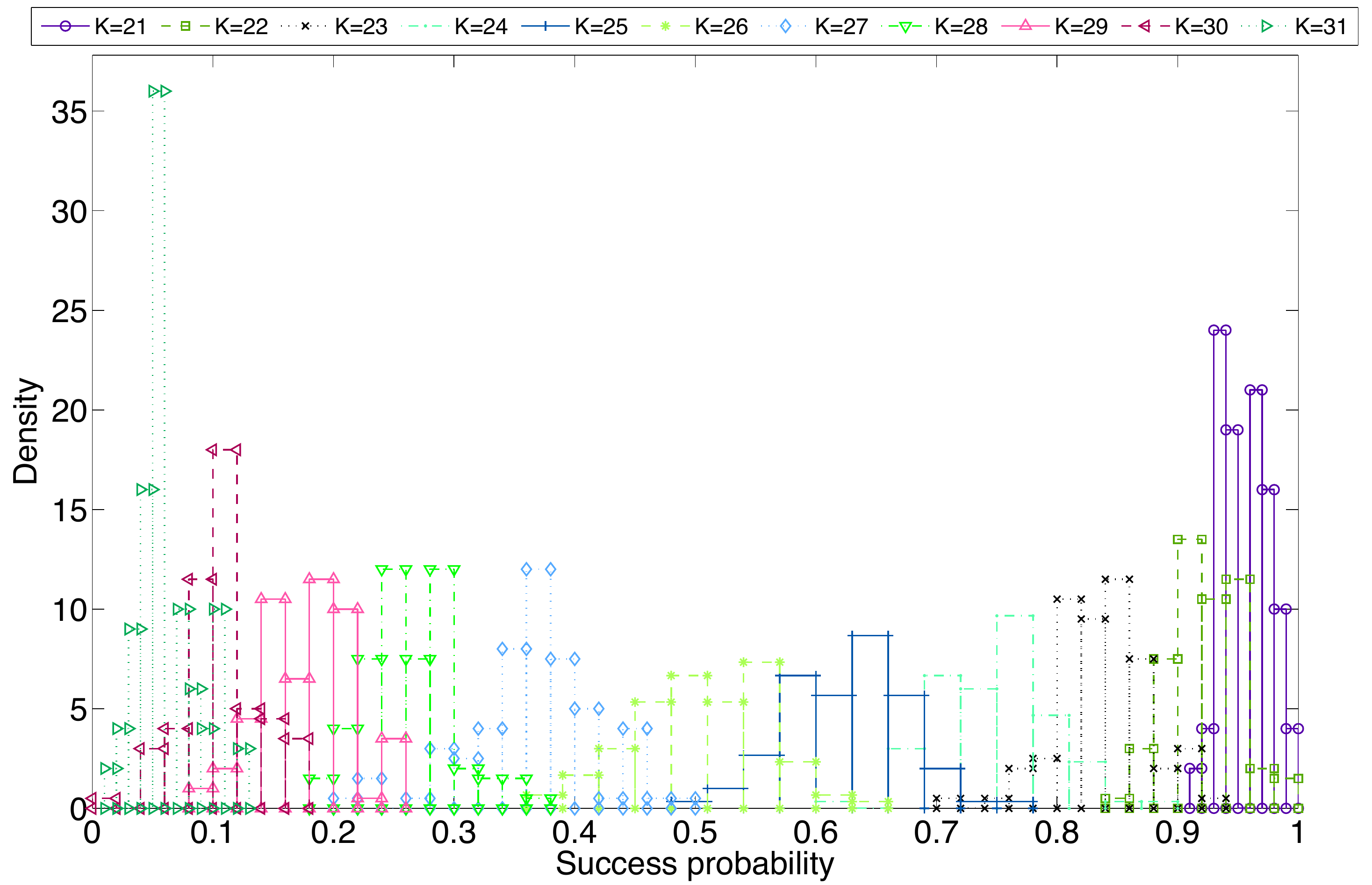}
\caption{Histograms of success probability for various $K$'s using Modified-CS~\cite{vaswani2010modified}. Success probability distribution for each $K$ was obtained with 100 different random signed spike vectors for $N=512$ and $M=100$. A single success probability for each signed spike vector was calculated with 100 experiments. Compared with the results with basis pursuit in Fig.~\ref{fig_three}, success probability is higher for a given sparsity $K$. Note also that 100 experiments resulted in higher variance for each $K$.}
\label{fig_modified}
\end{figure}

We also employed real-world environmental data sets obtained from wireless sensor network deployments~\cite{quer2012sensing}: humidity and temperature. In addition, audio data shown in Fig.~\ref{fig_one} was used for comparison as well. Random numbers representing the dynamic signal sparsity $K$ were drawn from the inverse Gaussian distribution ($f_{K}(k)=\mathrm{IG}(30,200)$) and we used this $K$ to randomly choose components sorted in decreasing order; other components were set to zero. Fig.~\ref{fig_experimental} displays the success probability of signal recovery follows the shape of Fig.~\ref{fig_efficacy}.

\begin{figure*}[!t]
\centering
\subfloat[humidity]{\includegraphics[width=0.33\linewidth]{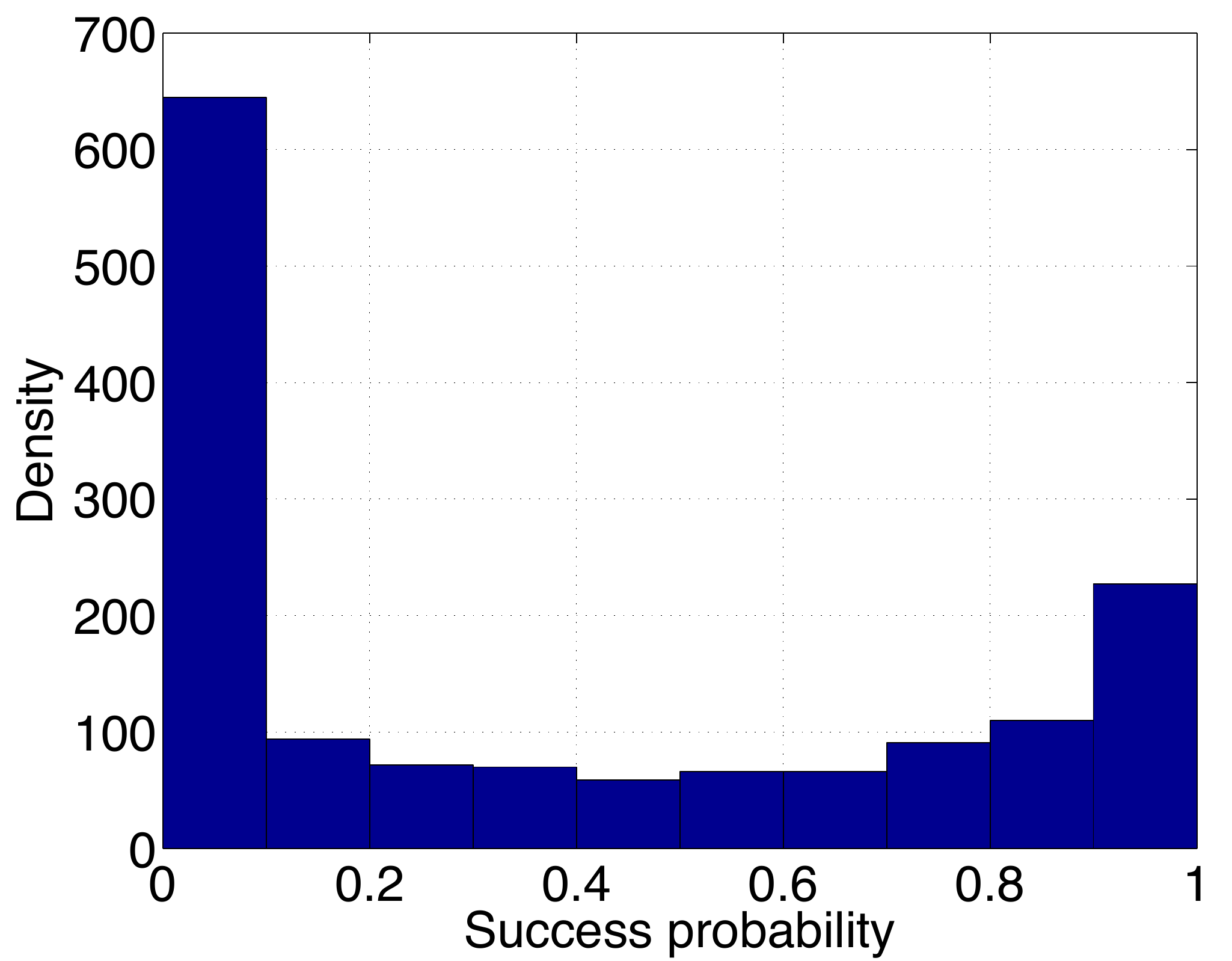}%
\label{fig_experimental_a}}
\hfil
\subfloat[temperature]{\includegraphics[width=0.33\linewidth]{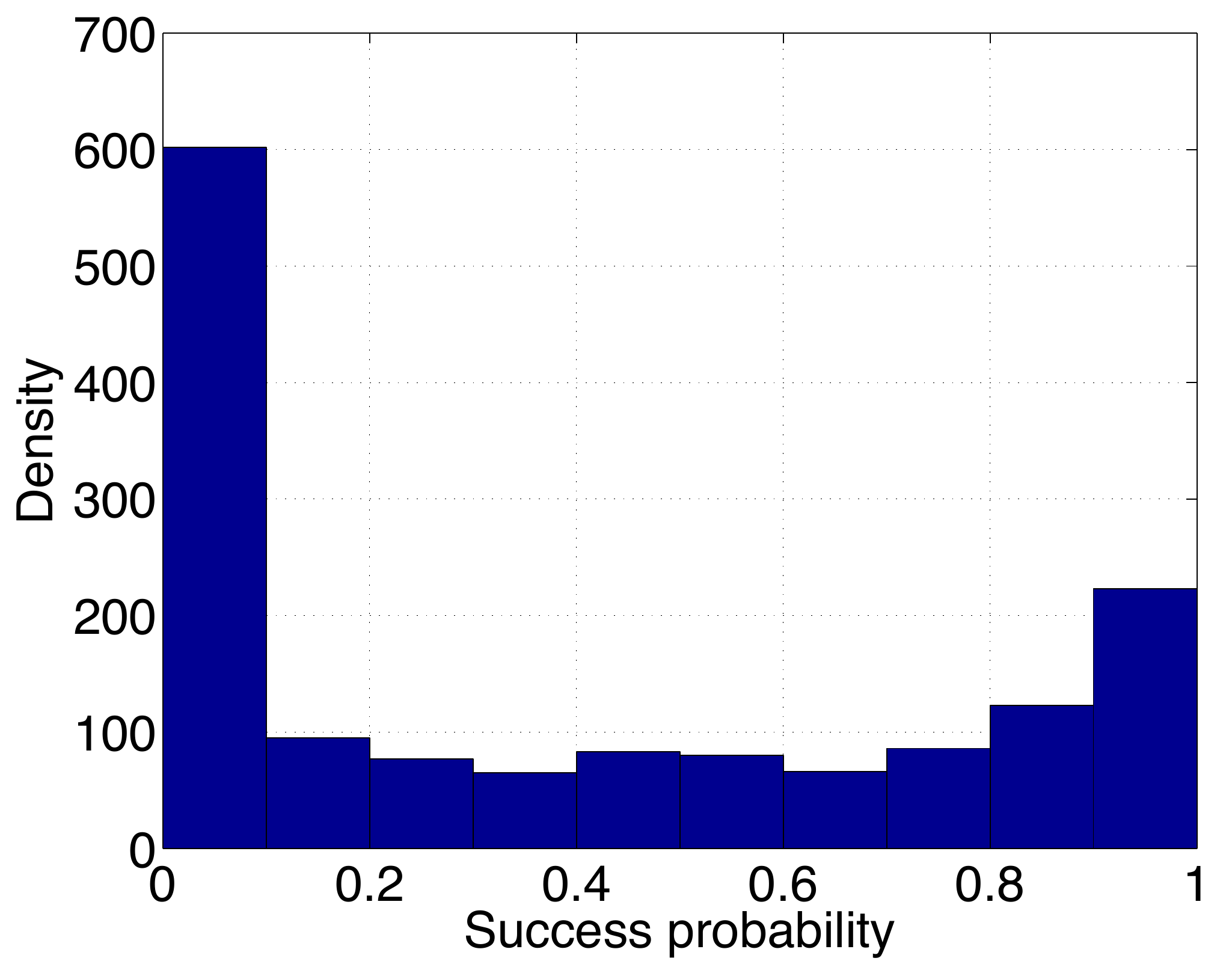}%
\label{fig_experimental_b}}
\hfil
\subfloat[audio]{\includegraphics[width=0.33\linewidth]{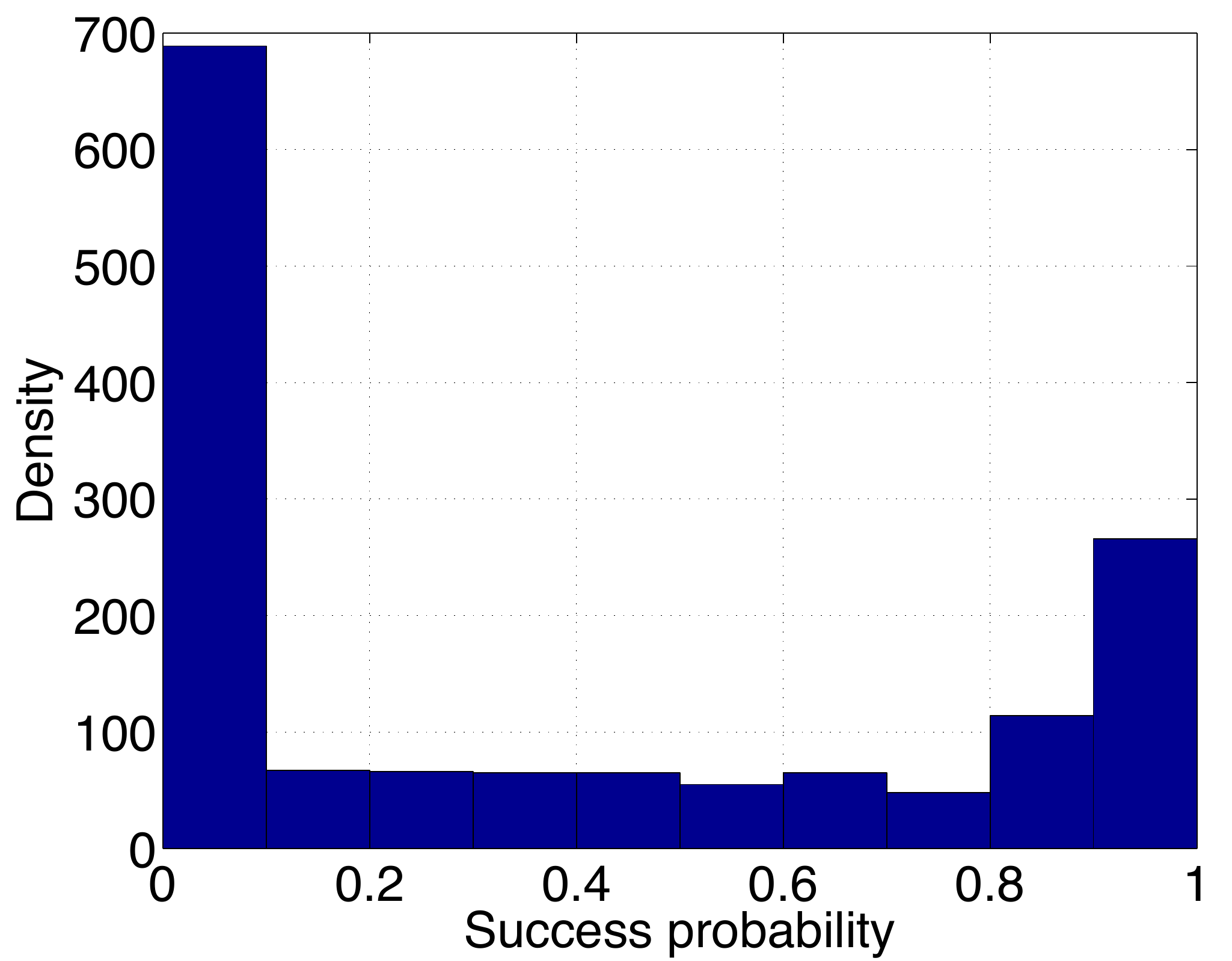}%
\label{fig_experimental_c}}
\caption{Histograms of success probability for (a) humidity data, (b) temperature data, and (c) audio data. Histogram was obtained with 1,500 random number generations (the inverse Gaussian distribution) to choose different signals and 100 different experiments for each signal with $N=512$ and $M=100$. Note that all histograms closely follow the shape of Fig.~\ref{fig_efficacy}.}
\label{fig_experimental}
\end{figure*}

\subsection{Recovery Quality}
When a signal is compressible and not exactly $K$-sparse, this signal is basically \emph{dense}. In Section~\ref{sec:quality}, we regarded the number of components included in the signal recovery as a varying quantity. We are interested in the general shape of this quantity in distribution. In order to verify Proposition~\ref{pro:1}, we performed experiments using real data sets as well as artificially generated random signed spikes.

We first provide results with real-world data sets to verify Proposition~\ref{pro:1}. Fig.~\ref{fig_gamma} displays the histograms of $K$, the number of components included in each signal recovery, which was obtained using the method explained in Section~\ref{sec:quality}. We can identify that Proposition~\ref{pro:1} actually holds here, as all distributions are skewed to the right. Furthermore, the distributions follow the gamma distribution, which is also natural since the gamma distribution has positive skewness, i.e., right tailed.

In addition, random signed spikes were artificially generated in different magnitudes at random locations and densified to perform random sampling. In particular, we considered an arithmetic sequence of length $50$ $(2,4,6,{\ldots},98,100)$, whose elements were placed at random locations in each vector. These signals are dense enough to be used for experiments because signal recovery always fails when $K>30$ in our case, as shown in Fig.~\ref{fig_two}. Fig.~\ref{fig_seven} displays the histogram of $K$ and the gamma distribution fitting, where we can again see that Proposition~\ref{pro:1} holds.

\begin{figure*}[!t]
\centering
\subfloat[humidity]{\includegraphics[width=0.33\linewidth]{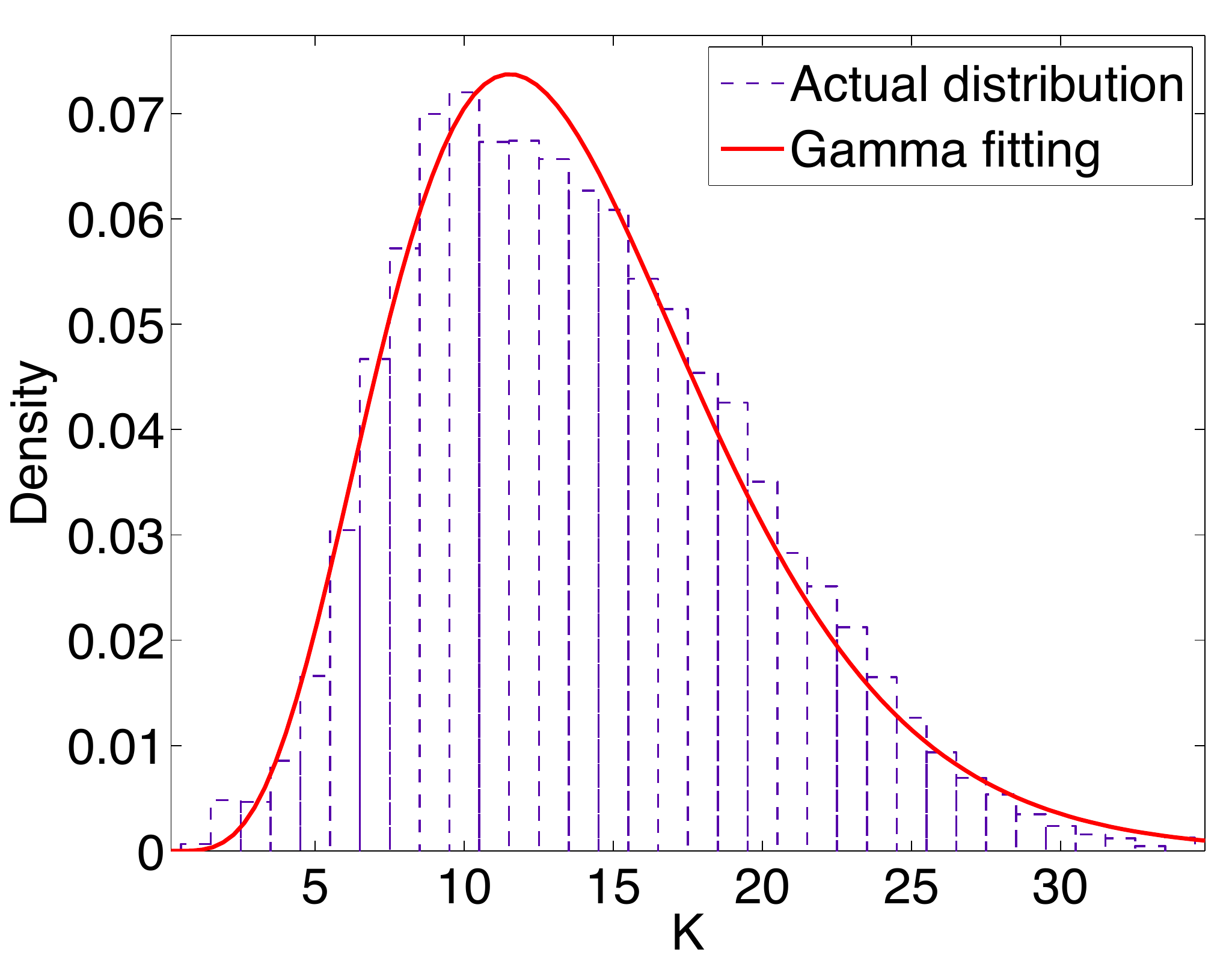}%
\label{fig_gamma_a}}
\hfil
\subfloat[temperature]{\includegraphics[width=0.33\linewidth]{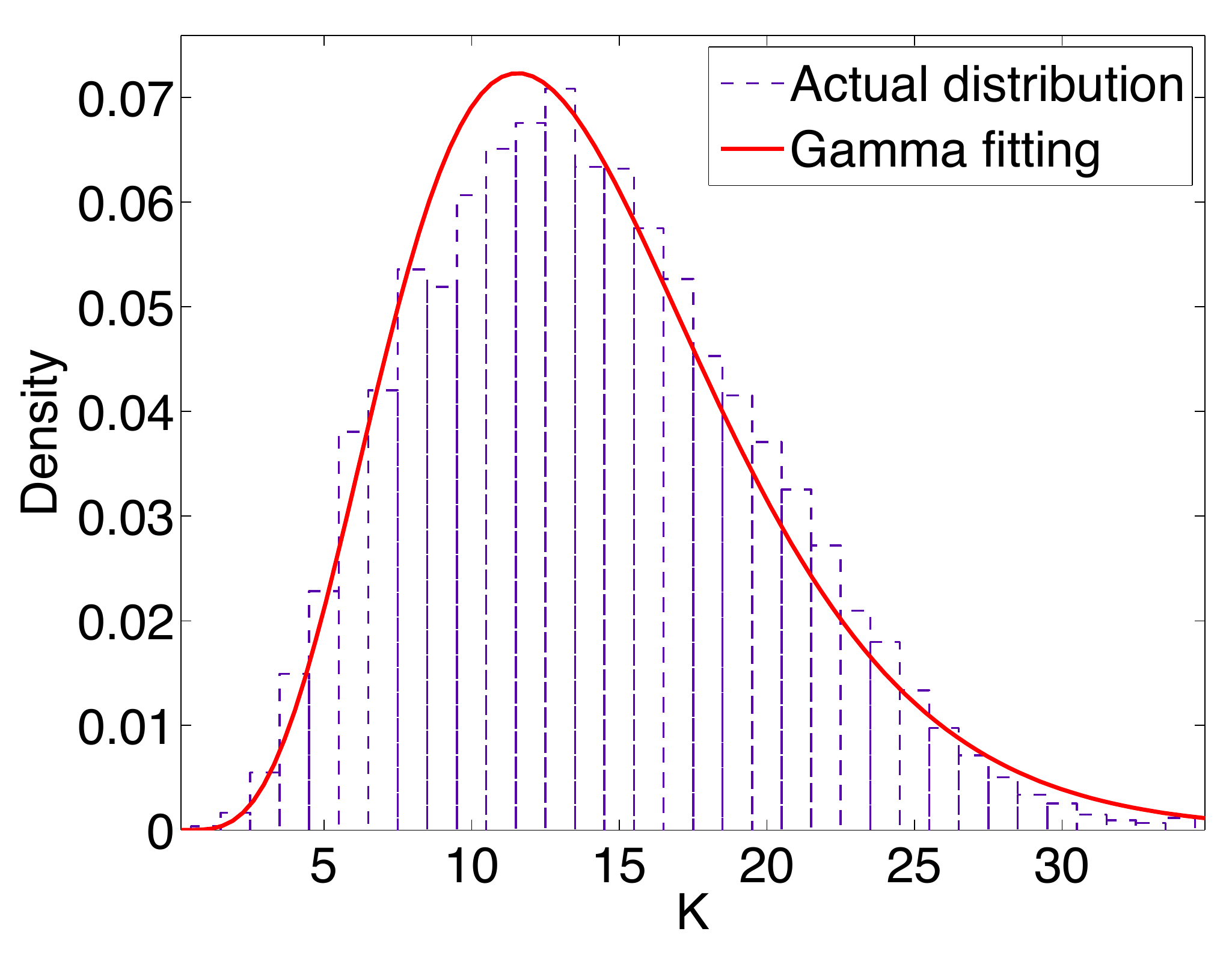}%
\label{fig_gamma_b}}
\hfil
\subfloat[audio]{\includegraphics[width=0.33\linewidth]{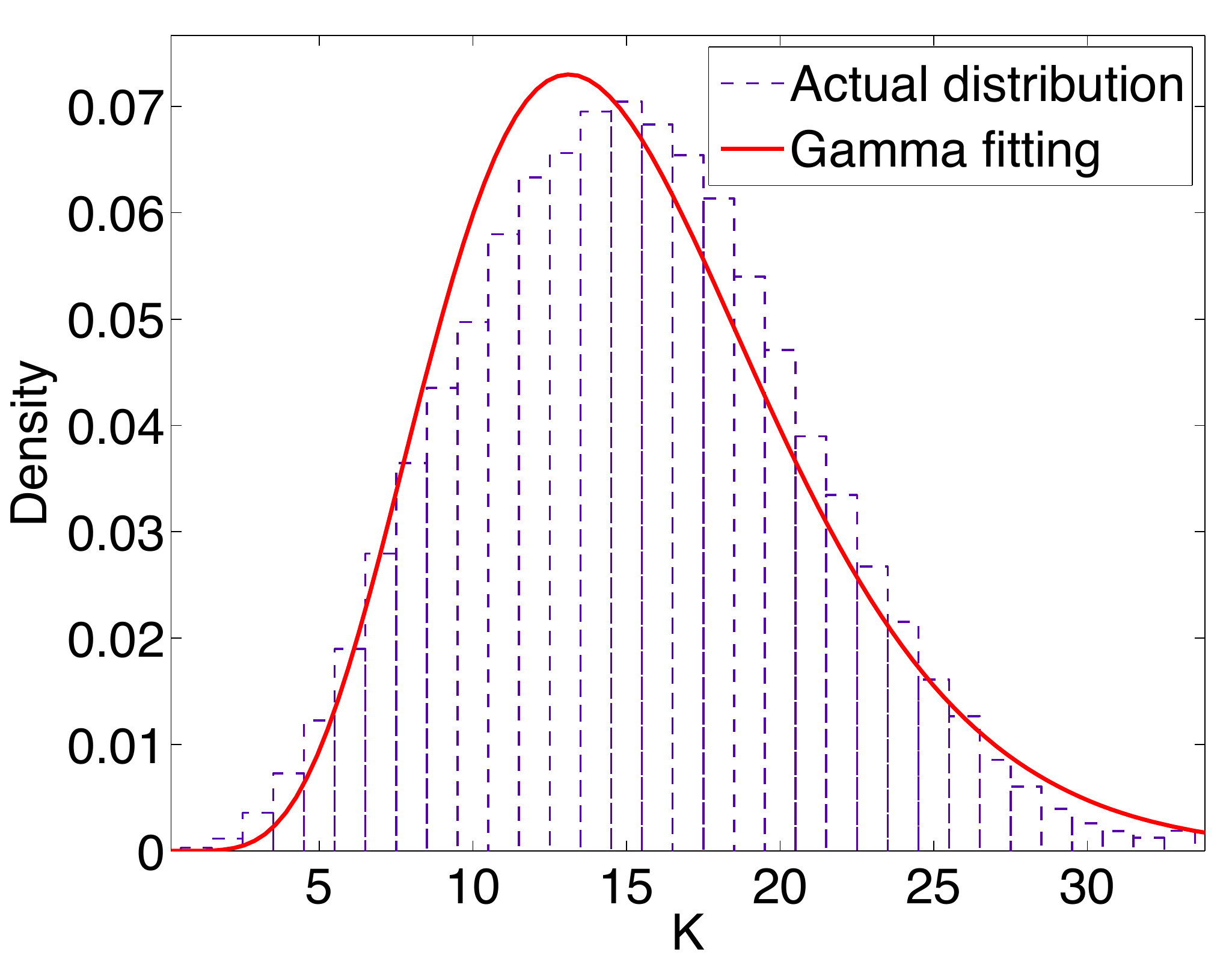}%
\label{fig_gamma_c}}
\caption{Distributions of $K$ fitted with gamma distributions for (a) humidity data ($\mathrm{Gamma}(5.69,2.45)$), (b) temperature data ($\mathrm{Gamma}(5.56,2.54)$), and (c) audio data ($\mathrm{Gamma}(6.92,2.21)$). Histograms were obtained with 34 different signals and 1,000 different experiments for each signal (a and b); with 153 different signals and 500 different experiments for each signal (c), with $N=512$ and $M=100$.}
\label{fig_gamma}
\end{figure*}

\begin{figure}[!t]
\centering
\includegraphics[width=0.6\columnwidth]{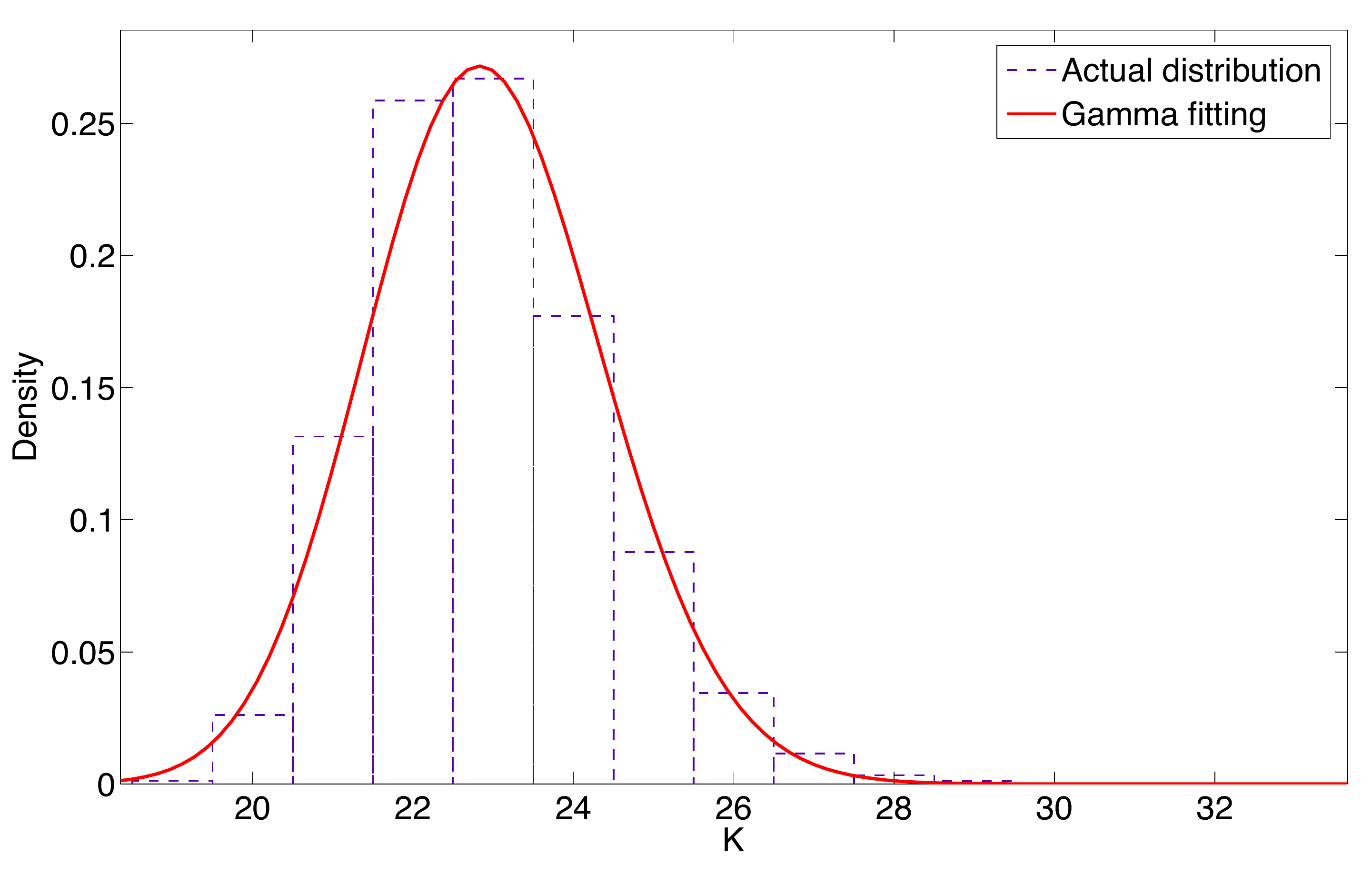}
\caption{Distribution of $K$ fitted with a gamma distribution $\mathrm{Gamma}(242.81,0.09)$, using the maximum likelihood estimation. Histogram was obtained with 300 different signals and 300 different experiments for each signal, with $N=512$ and $M=100$.}
\label{fig_seven}
\end{figure}

Furthermore, we analyze the $\ell_{2}$ error $E$ of signal recovery assuming $f_{K}(k)=\mathrm{Gamma}(\kappa,\theta)$ using (\ref{eq_seventeen}). In order to show its efficacy, we compared the solutions of (\ref{eq_seventeen}) with real data sets. For humidity data, $E=94.3533$ while the average $\ell_{2}$ norm of data is $564.8585$; for temperature data, $E=75.5441$ while the average $\ell_{2}$ norm of data is $627.8038$; and for audio data, $E=5.0979$ while the average $\ell_{2}$ norm of data is $1.5866$. Apart from the case of audio data, (\ref{eq_seventeen}) provides useful estimators for the upper bound of amount of error during recovery. It should be noted that this bound is rather loose due to a large constant $G$ in (\ref{eq_sixteen}), which could be improved with a less conservative $G$.

\section{Conclusion} \label{sec:conclusion}
We have presented a new theoretical CS framework in random sampling which handles uncertainty in signal recovery from a new perspective. The success probability of signal recovery in random sampling was investigated when the signal sparsity can vary with an insufficient number of measurements. The success probability analysis in the existing CS framework was shown to be incapable of reflecting actual success probability by both theoretical analysis and experiments. On the contrary, our recovery success model could closely reflect actual success probability.

We also considered signals which cannot be exactly represented with sparse representations, where we could alternatively view the number of components included in the signal recovery as a varying quantity. This quantity was shown by both theoretical analysis and experiments to follow a right-tailed distribution such as the gamma distribution. We provided an error analysis for these signals.

\bibliographystyle{IEEEtran}
\bibliography{DELEE_Improving_CS}

\begin{thebibliography}{10}
\providecommand{\url}[1]{#1}
\csname url@samestyle\endcsname
\providecommand{\newblock}{\relax}
\providecommand{\bibinfo}[2]{#2}
\providecommand{\BIBentrySTDinterwordspacing}{\spaceskip=0pt\relax}
\providecommand{\BIBentryALTinterwordstretchfactor}{4}
\providecommand{\BIBentryALTinterwordspacing}{\spaceskip=\fontdimen2\font plus
\BIBentryALTinterwordstretchfactor\fontdimen3\font minus
  \fontdimen4\font\relax}
\providecommand{\BIBforeignlanguage}[2]{{%
\expandafter\ifx\csname l@#1\endcsname\relax
\typeout{** WARNING: IEEEtran.bst: No hyphenation pattern has been}%
\typeout{** loaded for the language `#1'. Using the pattern for}%
\typeout{** the default language instead.}%
\else
\language=\csname l@#1\endcsname
\fi
#2}}
\providecommand{\BIBdecl}{\relax}
\BIBdecl

\bibitem{bajwa2006compressive}
W.~Bajwa, J.~Haupt, A.~Sayeed, and R.~Nowak, ``Compressive wireless sensing,''
  in \emph{Proc. IPSN}, 2006, pp. 134--142.

\bibitem{ji2007bayesian}
S.~Ji and L.~Carin, ``Bayesian compressive sensing and projection
  optimization,'' in \emph{Proc. ICML}, 2007, pp. 377--384.

\bibitem{Seeger2008ICML}
M.~W. Seeger and H.~Nickisch, ``Compressed sensing and bayesian experimental
  design,'' in \emph{Proc. ICML}, 2008, pp. 912--919.

\bibitem{luo2009compressive}
C.~Luo, F.~Wu, J.~Sun, and C.~W. Chen, ``Compressive data gathering for
  large-scale wireless sensor networks,'' in \emph{Proc. MobiCom}, 2009, pp.
  145--156.

\bibitem{hsu2009NIPS}
D.~Hsu, S.~Kakade, J.~Langford, and T.~Zhang, ``Multi-label prediction via
  compressed sensing.'' in \emph{Proc. NIPS}, 2009, pp. 772--780.

\bibitem{Lopes13}
M.~Lopes, ``Estimating unknown sparsity in compressed sensing,'' in \emph{Proc.
  ICML}, 2013, pp. 217--225.

\bibitem{malioutov2013exact}
D.~Malioutov and K.~Varshney, ``Exact rule learning via boolean compressed
  sensing,'' in \emph{Proc. ICML}, 2013, pp. 765--773.

\bibitem{baraniuk2007compressive}
R.~G. Baraniuk, ``Compressive sensing [lecture notes],'' \emph{{IEEE} Signal
  Process. Mag.}, vol.~24, no.~4, pp. 118--121, Jul. 2007.

\bibitem{foucart2013mathematical}
S.~Foucart and H.~Rauhut, \emph{A Mathematical Introduction to Compressive
  Sensing}.\hskip 1em plus 0.5em minus 0.4em\relax Springer, 2013.

\bibitem{lee2014big}
D.~Lee and J.~Choi, ``Low complexity sensing for big spatio-temporal data,'' in
  \emph{Proc. BigData}, 2014, pp. 323--328.

\bibitem{sejdinovic2010Allerton}
D.~Sejdinovic, C.~Andrieu, and R.~Piechocki, ``Bayesian sequential compressed
  sensing in sparse dynamical systems,'' in \emph{Proc. Allerton}, 2010, pp.
  1730--1736.

\bibitem{Shahrasbi2011CISS}
B.~Shahrasbi, A.~Talari, and N.~Rahnavard, ``{TC}-{CSBP}: Compressive sensing
  for time-correlated data based on belief propagation,'' in \emph{Proc. CISS},
  2011, pp. 1--6.

\bibitem{vaswani2010modified}
N.~Vaswani and W.~Lu, ``Modified-{CS}: Modifying compressive sensing for
  problems with partially known support,'' \emph{{IEEE} Trans. Signal
  Process.}, vol.~58, no.~9, pp. 4595--4607, Sep. 2010.

\bibitem{Ziniel2013TSP}
J.~Ziniel and P.~Schniter, ``Dynamic compressive sensing of time-varying
  signals via approximate message passing,'' \emph{{IEEE} Trans. Signal
  Process.}, vol.~61, no.~21, pp. 5270--5284, Nov. 2013.

\bibitem{ganguli2010NIPS}
S.~Ganguli and H.~Sompolinsky, ``Short-term memory in neuronal networks through
  dynamical compressed sensing,'' in \emph{Proc. NIPS}, 2010, pp. 667--675.

\bibitem{Malioutov2010JSTS}
D.~M. Malioutov, S.~R. Sanghavi, and A.~S. Willsky, ``Sequential compressed
  sensing,'' \emph{{IEEE} J. Sel. Top. Signal Process.}, vol.~4, no.~2, pp.
  435--444, Apr. 2010.

\bibitem{candes2008introduction}
E.~J. Cand{\`e}s and M.~B. Wakin, ``An introduction to compressive sampling,''
  \emph{{IEEE} Signal Process. Mag.}, vol.~25, no.~2, pp. 21--30, Mar. 2008.

\bibitem{boyd2004convex}
S.~Boyd and L.~Vandenberghe, \emph{Convex Optimization}.\hskip 1em plus 0.5em
  minus 0.4em\relax Cambridge University Press, 2004.

\bibitem{pati1993orthogonal}
Y.~C. Pati, R.~Rezaiifar, and P.~S. Krishnaprasad, ``Orthogonal matching
  pursuit: Recursive function approximation with applications to wavelet
  decomposition,'' in \emph{Proc. ACSSC}, 1993, pp. 40--44.

\bibitem{blumensath2008iterative}
T.~Blumensath and M.~E. Davies, ``Iterative thresholding for sparse
  approximations,'' \emph{J. Fourier Anal. Appl.}, vol.~14, no. 5-6, pp.
  629--654, Dec. 2008.

\bibitem{ZinielRS12}
J.~Ziniel, S.~Rangan, and P.~Schniter, ``A generalized framework for learning
  and recovery of structured sparse signals,'' in \emph{Proc. {IEEE}
  Statistical Signal Processing Workshop}, 2012, pp. 325--328.

\bibitem{durbin1998biological}
R.~Durbin, S.~R. Eddy, A.~Krogh, and G.~Mitchison, \emph{Biological Sequence
  Analysis: Probabilistic Models of Proteins and Nucleic Acids}.\hskip 1em plus
  0.5em minus 0.4em\relax Cambridge University Press, 1998.

\bibitem{lee2015learning}
D.~Lee and J.~Choi, ``Learning compressive sensing models for big
  spatio-temporal data,'' in \emph{Proc. SDM}, 2015, pp. 667--675.

\bibitem{baraniuk2010model}
R.~G. Baraniuk, V.~Cevher, M.~F. Duarte, and C.~Hegde, ``Model-based
  compressive sensing,'' \emph{{IEEE} Trans. Inf. Theory}, vol.~56, no.~4, pp.
  1982--2001, Apr. 2010.

\bibitem{miller1980bayesian}
R.~B. Miller, ``Bayesian analysis of the two-parameter gamma distribution,''
  \emph{Technometrics}, vol.~22, no.~1, pp. 65--69, Feb. 1980.

\bibitem{Fink97acompendium}
D.~Fink, ``A compendium of conjugate priors,'' 1997.

\bibitem{chen1998atomic}
S.~S. Chen, D.~L. Donoho, and M.~A. Saunders, ``Atomic decomposition by basis
  pursuit,'' \emph{SIAM J. Sci. Comput.}, vol.~20, no.~1, pp. 33--61, Jan.
  1998.

\bibitem{quer2012sensing}
G.~Quer, R.~Masiero, G.~Pillonetto, M.~Rossi, and M.~Zorzi, ``Sensing,
  compression, and recovery for {WSN}s: Sparse signal modeling and monitoring
  framework,'' \emph{{IEEE} Trans. Wireless Commun.}, vol.~11, no.~10, pp.
  3447--3461, Oct. 2012.

\end{thebibliography}
%

\end{document}